\documentclass[11pt, reqno, a4paper]{amsart}
\usepackage[table,xcdraw]{xcolor}
\usepackage{braket}
\usepackage{graphicx}
\usepackage{tikz}
\usepackage[backend=biber,style=nature]{biblatex}
\usepackage{amssymb,amsthm}
\usepackage[margin=1in]{geometry}
\usepackage[colorlinks = true, linkcolor = blue, urlcolor  = blue, citecolor = red]{hyperref}
\usepackage{subfig}
\renewcommand{\epsilon}{\varepsilon}
\renewcommand{\phi}{\varphi}
\newcommand{\C}{\mathbb{C}}
\newcommand{\R}{\mathbb{R}}

\newcommand{\M}{\mathcal{M}}
\DeclareMathOperator{\Tr}{Tr}

\newcommand{\ketbra}[2]{|#1\rangle\langle#2|} 
\newtheorem{theorem}{Theorem}[section]
\newtheorem{definition}[theorem]{Definition}
\newtheorem{proposition}[theorem]{Proposition}
\newtheorem{corollary}[theorem]{Corollary}
\newtheorem{lemma}[theorem]{Lemma}
\newtheorem{remark}[theorem]{Remark}
\newtheorem{example}[theorem]{Example}

\DeclareMathAlphabet\mathbfcal{OMS}{cmsy}{b}{n}

\begin{document}

	\title{The compatibility dimension of quantum measurements}

\author{Faedi Loulidi}
	
\author{Ion Nechita}
\email{$\{$loulidi,nechita$\}$@irsamc.ups-tlse.fr}
\address{Laboratoire de Physique Th\'eorique, Universit\'e de Toulouse, CNRS, UPS, France}	
	
	\begin{abstract}
		We introduce the notion of compatibility dimension for a set of quantum measurements: it is the largest dimension of a Hilbert space on which the given measurements are compatible. In the Schr\"odinger picture, this notion corresponds to testing compatibility with ensembles of quantum states supported on a subspace, using the incompatibility witnesses of Carmeli, Heinosaari, and Toigo. We provide several bounds for the compatibility dimension, using approximate quantum cloning or algebraic techniques inspired by quantum error correction. We analyze in detail the case of two orthonormal bases, and, in particular, that of mutually unbiased bases. 
	\end{abstract}
	
	\date{\today}
	
	\maketitle
	
	\section{Introduction}

	The process of measurement in quantum mechanics has many properties differentiating it from what one encounters in classical theories. First of all, Born's rule states that the outcome of a quantum measurement is probabilistic, quantum theory predicting only the probability distribution of possible outcomes. Heisenberg's uncertainty principle gives a lower bound on the joint precision with which values can be attributed to general quantum observables. Closely related to the latter is the notion of \emph{quantum incompatibility}: there exist quantum measurements that cannot be performed simultaneously on an unknown quantum state. Incompatibility of quantum measurements has received a lot of attention from both theorists (as a signature of quantumness) and experimentalists (mainly due to the relation to Bell non-locality \cite{fine1982hidden,wolf2009measurements,brunner2014bell}).

	For a pair of incompatible quantum measurements, it is well known that adding enough noise renders them compatible \cite{busch1996quantum,busch2013comparing}. This has been a very fruitful direction of research, see the recent review \cite{designolle2019incompatibility} and the connection to free spectrahedra \cite{bluhm2018joint,bluhm2020compatibility}. In this work, we study a different approach to the same problem of making measurements compatible, by \emph{dimension reduction}. This can be understood in two equivalent ways: 
	\begin{itemize}
		\item taking corners of the POVM elements (Heisenberg picture)
		\item restricting the sets of quantum states to a subspace (Schr\"odinger picture).
	\end{itemize}
	
	We introduce a measure of incompatibility of measurements from this perspective: the \emph{compatibility dimension} of a tuple of POVMs $A^{(1)}, \ldots, A^{(g)}$ is the largest Hilbert space dimension $r$ for which there \emph{exists} an isometry $V : \C^r \to \C^d$ such that the reduced POVMs $V^*A^{(1)}V, \ldots, V^* A^{(g)}V$ are compatible, see Definition \ref{def:R-Rbar}. Similarly, we define the \emph{strong compatibility dimension} of a tuple of measurements as the largest dimension $r$ for which \emph{all} isometries $V:\C^r \to \C^d$ reduce the POVMs to a compatible tuple.
	
    We study different examples and fundamental properties of these newly defined quantities. Using analytic and algebraic techniques, we prove several bounds in the most relevant cases. For the case of two von Neumann measurements, we relate the compatibility dimension to a geometric quantity encoding the relative position of the vectors of the two bases. For two noisy mutually unbiased bases, we show that, for some particular values of the noise parameters, dimensionality reduction renders incompatible measurements compatible. To do so, we prove along the way a generalization of a compatibility criterion \cite{heinosaari2014maximally} coming from quantum cloning. 
    We relate these dimensions to the notion of incompatibility witnesses introduced in \cite{carmeli2018state,carmeli2019quantum}, using the measurement / state duality. We use algebraic techniques inspired from the theory of quantum error correction to prove very general lower bounds on the compatibility dimension. Finally, we consider spin systems coming from Clifford algebras as an illuminating example. 
    
    The newly introduced measure, the \emph{compatibility dimension} of a tuple of quantum measurements, sheds light on the complex phenomenon of quantum incompatibility. It is a discrete measure of incompatibility: compatible POVMs have maximal compatibility dimension (equal to that of the ambient Hilbert space), while smaller compatibility dimensions indicate a higher robustness of incompatibility. We provide a plethora of results regarding this measure, of both analytical and algebraic flavor, focusing on important classes of POVMs, such as noisy mutually unbiased von Neumann measurements. We leave a certain number of questions regarding the compatibility dimension open, and hope that our work will stimulate further research in this direction. 
    
    Our paper is organized as follows. In Section \ref{sec:compatibility} we recall the main definitions and the basic properties of quantum measurements, focusing on the notion of compatibility. We present in Section \ref{sec:cloning} a generalization of a compatibility criterion using asymmetric cloning. Section \ref{sec:comp-dim} contains the main definitions of the paper, that of the (strong) compatibility dimension. We switch to the Schr\"odinger picture in Section \ref{sec:IW}, relating the compatibility dimension to incompatibility witnesses and discrimination of state super-ensembles. Sections \ref{sec:2-bases} and \ref{sec:MUB} are devoted to two important examples: von Neumann measurements and (noisy) mutually unbiased bases. In Section \ref{sec:algebraic} we use techniques inspired by quantum error correction to provide very general lower bounds for the compatibility dimension. Finally, we study spin systems in Section \ref{sec:spin-systems}, obtaining lower bounds for the strong compatibility dimension. We conclude with a list of open questions and directions for further research. 
	
\section{Compatibility of quantum measurements}\label{sec:compatibility}

We gather in this section the main definitions and basic facts from the theory of quantum measurements. In quantum mechanics, to quantum systems we associate a complex Hilbert space $\mathcal H$. In this paper, we shall focus on finite dimensional Hilbert spaces, so we shall write $\mathcal H \cong \C^d$ for a positive integer $d$, the number of degrees of freedom of the quantum system. We denote by $\mathcal M_d$ the vector space of $d \times d$ complex matrices.  The \emph{states} of a quantum system are mathematically modelled by \emph{density matrices}
$$\M^{1,+}_d := \{ \rho \in \M_d \, : \, \rho \geq 0 \text{ and } \Tr \rho = 1\},$$
where $\rho \geq 0$ means that the matrix $\rho$ is positive semidefinite (i.e.~$\rho$ is self-adjoint and has non-negative eigenvalues). 

The measurement process is modelled in quantum mechanics by \emph{observables}. This formalism allows to obtain the probability distribution of the possible outcomes, as well as the state of the system after the measurement (the \emph{wave function collapse}). In this work, we are interested in the probabilities of outcomes only, so we shall use the framework of POVMs. We write $[n]:= \{1, 2, \ldots, n\}$.

\begin{definition}
    A \emph{positive operator valued measure} (POVM) on $\M_d$ is a tuple $A=(A_1, \ldots, A_k)$ of self-adjoint operators from $\M_d$ which are positive semidefinite and sum up to the identity:
    $$\forall i \in [k], \quad A_i \geq 0 \qquad \text{ and } \qquad \sum_{i=1}^k A_i = I_d.$$
    When measuring a POVM $A$ on a quantum system in state $\rho$, we obtain a random outcome
    $$\forall i \in [k], \qquad \mathbb P(\text{outcome} = i) = \Tr[\rho A_i].$$
\end{definition}

The properties of the POVM operators $A_i$ (called \emph{quantum effects}) ensure that the vector $\left(\Tr[\rho A_i]\right)_{i=1}^k$ is a probability vector. Note that this mathematical formalism does not account for what happens with the quantum particle after the measurement; we say that the particle is destroyed in the process of measurement, see Figure \ref{fig:measurement}. 

\begin{figure}[htb!]
    \centering
    \includegraphics{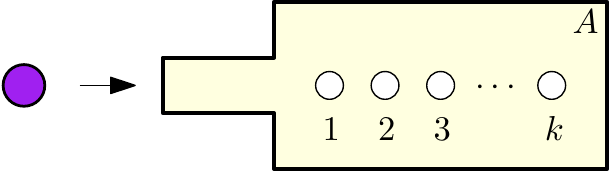}\qquad\qquad\qquad\qquad 
    \includegraphics{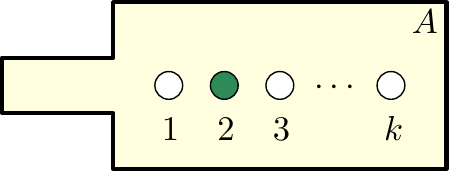}
    \caption{Diagrammatic representation of a quantum measurement. Left: a quantum particle enters a measurement apparatus. Right: after the measurement is performed, the particle is destroyed, and the apparatus displays the classical outcome (here, $2$).}
    \label{fig:measurement}
\end{figure}

An important class of POVMs are \emph{von Neumann measurements}, where $A_i = \ketbra{a_i}{a_i}$, $i \in [d]$, for an orthonormal basis $\{\ket{a_i}\}_{i=1}^d$ of $\C^d$. On the other side of the spectrum, there are \emph{trivial} POVMs, where $B_j = q_j I_d$, for some probability vector $q=(q_1, \ldots, q_k)$. Note that for trivial POVMs, the outcome probabilities are given by the vector $q$, independently of the quantum state $\rho$ that is being measured. The special case of equi-probability $q_j = 1/k$ will be of interest in this paper: we define the notion of noisy POVMs, with respect to the \emph{random} or \emph{uniform} noise model (see \cite{designolle2019incompatibility}). 
\begin{definition}
For a POVM $A$ and a parameter $t \in [0,1]$, we define the \emph{noisy version} $\mathcal N_t[A]$ of $A$ by
$$\mathcal N_t[A]_i = tA_i + (1-t)\frac{I_d}{k},$$
where $k$ is the number of outcomes of $A$. In other words, $\mathcal N_t[A]$ is the convex combination, with weight $t$, between $A$ and the uniform trivial POVM $(I_d/k,\ldots, I_d/k)$.

Similarly, for $g$-tuples of POVMs $\mathbf A = (A^{(1)}, \ldots, A^{(g)})$, we define
$$\mathcal N_{\mathbf t}[\mathbf A] = (\mathcal N_{t_1}[A^{(1)}], \ldots, \mathcal N_{t_g}[A^{(g)}]),$$
for a vector $\mathbf t \in [0,1]^g$. If the vector $\mathbf  t$ is constant, $\mathbf t = (t, t, \ldots, t)$, we write $\mathcal N_{t}[\mathbf A] := \mathcal N_{\mathbf t}[\mathbf A]$.
\end{definition}
Note that in the definition above, we allow POVMs having possibly different number of outcomes. 

Of central importance in this work will be the following notion. 
\begin{definition}\label{def:reduced-POVMs}
    Given an isometry $V : \mathbb C^r \to \mathbb C^d$ and a POVM $A = (A_1, \ldots, A_k)$ on $\mathcal M_d$, we define the \emph{reduced POVM} on $\mathcal M_r$ 
    $$V^*AV := (V^*A_1V, \ldots, V^*A_kV).$$
\end{definition}
We record here the following result, which will be used later in the paper. 
\begin{lemma}\label{lem:noise-isometry}
For a POVM $A$ on $\mathcal M_d$ and an isometry $V : \mathbb C^r \to \mathbb C^d$, we have 
$$V^* \mathcal N_t[A] V = \mathcal N_t[V^*AV].$$
\end{lemma}
\begin{proof}
This simple fact follows from the special type of noise we use: 
$$V^*\mathcal N_t[A]_i V = tV^*A_iV + (1-t)\frac{V^*I_dV}{k} = tV^*A_iV + (1-t)\frac{I_r}{k} = \mathcal N_t[V^*AV]_i.$$
\end{proof}

\medskip

We introduce now the notion of \emph{compatibility} for POVMs, which is central to this paper. Physically, this notion is motivated by the following scenario. Suppose we want to measure two different physical quantities (modelled by two POVMs $A$ and $B$) on a given quantum particle in a state $\rho$. Since the particle is destroyed after performing a given measurement, we cannot measure simultaneously $A$ and $B$. However, measuring $A$ and $B$ on $\rho$ can be simulated by measuring a different POVM $C$, and then \emph{classically} post-processing the output of $C$ to a pair of outcomes $(i,j)$ for $A$, respectively $B$, see Figure \ref{fig:compatibility}. Famously, there are pairs of POVMs $A$ and $B$ for which there is no such $C$, like the position and momentum operators of a particle in one dimension: it is impossible to attribute an exact value to both position and momentum observables at the same time. 

\begin{figure}[htb!]
    \centering
    \includegraphics[width=0.9\textwidth]{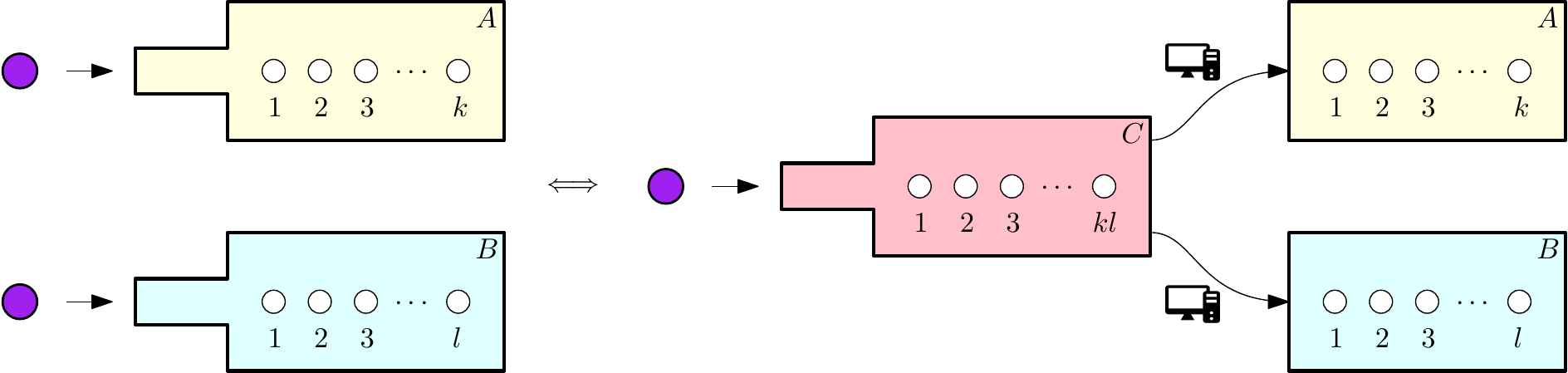}
    \caption{The simultaneous measurement of $A$ and $B$ is simulated by the measurement of $C$ on a single copy of the quantum particle, followed by a classical post-processing of the output of $C$.}
    \label{fig:compatibility}
\end{figure}

Mathematically, we have the following important definition, see, e.g., the excellent review paper \cite{heinosaari2016invitation}.

\begin{definition}
    Two POVMs $A=(A_1, \ldots, A_k)$, $B=(B_1, \ldots, B_l)$ on $\M_d$ are called \emph{compatible} if there exists a POVM $C=(C_{11}, \ldots, C_{kl})$ on $\M_d$ such that $A$ and $B$ are its respective \emph{marginals}:
    \begin{align*}
        \forall i \in [k], \qquad A_i &= \sum_{j=1}^l C_{ij}\\
        \forall j \in [l], \qquad B_j &= \sum_{i=1}^k C_{ij}.
    \end{align*}
    If this is the case, the POVM $C$ is called a \emph{joint measurement} of $A$ and $B$. 

    More generally, a $g$-tuple of POVMs $\mathbf A = (A^{(1)}, \ldots, A^{(g)})$ is called compatible if there exists a POVM $C$ with outcome set $[k_1] \times \cdots \times [k_g]$ such that, for all $x \in [g]$, the POVM $A^{(x)}$ is the $x$-th marginal of $C$: 
    \begin{align*}
        \forall i_x \in [k_x], \qquad A^{(x)}_{i_x} &= \sum_{i_1 = 1}^{k_1} \cdots \sum_{i_{x-1} = 1}^{k_{x-1}}\sum_{i_{x+1} = 1}^{k_{x+1}} \cdots \sum_{i_g = 1}^{k_g} C_{i_1i_2 \cdots i_g}\\
        &= \sum_{\substack{\mathbf j \in [k_1] \times \cdots \times [k_g]\\ j_x = i_x}} C_{\mathbf j}.
    \end{align*}
\end{definition}

There is a lot of literature about the compatibility relation for quantum measurements, see \cite{heinosaari2016invitation}. Let us just mention here that in the case of two POVMs $A, B$ where at least one of them is projective (i.e.~the effect operators are projections), compatibility is equivalent to commutativity $[A_i,B_j] = 0$, for all $(i,j) \in [k]\times[l]$, see \cite[Proposition 8]{heinosaari2008notes}.

Given a pair of incompatible POVMs $A$ and $B$, it is always possible to render them compatible by mixing in some noise: 
$$\forall A,B \text{ POVMs}, \quad \mathcal N_{1/2}[A] \text{ and } \mathcal N_{1/2}[B] \text{ are compatible.}$$
Whether smaller amounts of noise suffice to render arbitrary POVMs compatible \cite{busch2013comparing} is a very important ongoing research question, see \cite{designolle2019incompatibility} for a recent review, and \cite{bluhm2018joint,bluhm2020compatibility} for a novel approach based on free spectrahedra. In this work, we introduce and study a different method of achieving compatibility of POVMs: instead of mixing in noise, we \emph{reduce their dimension}. 

\section{Compatibility criteria from asymmetric cloning}\label{sec:cloning}

We present now a generalization of the compatibility criterion from \cite{heinosaari2014maximally} to the case of several POVMs and asymmetric noise parameters. We obtain a necessary condition for the compatibility of a tuple of POVMs, which is in a sense dual to the asymmetric cloning problem. 

First, let us recall some basic facts about (asymmetric) cloning. It was shown that in quantum mechanics we cannot make exact copies of an arbitrary unknown quantum state \cite{wootters1982single}. This fact was formulated as \emph{the no-cloning theorem}, which is one of the fundamental differences between the classical and the quantum worlds. To precisely state a quantitative version of this fundamental fact, let us recall the basic definitions of completely positive maps and quantum channels; we refer the reader interested in background material on quantum information theory to the monograph \cite{watrous2018theory}.

	\begin{definition} A linear map $\Phi:\M_d \to \M_D$ is called \emph{completely positive} if for all $K\geq1$ and $X\in\M_d\otimes\M_K$, we have
    $$X\geq 0 \implies [\Phi\otimes \mathrm{id}_K](X)\geq 0,$$
    where $\mathrm{id}_K$ denotes the identity map. If, moreover, the map $\Phi$ is trace preserving
    $$\forall Y \in \M_d, \qquad \Tr \Phi(Y) = \Tr Y,$$
    then $\Phi$ is called a \emph{quantum channel}. 
\end{definition}

The no-cloning theorem can be precisely formulated as follows: for any number of clones $g \geq 2$, there is no quantum channel $\Phi : \M_d \to \M_d^{\otimes g}$ with the property that
$$\forall \rho \in \M_d^{1,+}, \, \forall j \in [g], \qquad \Tr_{[g] \setminus \{j\}}\Phi(\rho) = \rho.$$
The relation above means that there is no universal $1 \to g$ quantum cloner such that the $j$-th marginal of the output is equal to the input, for all $j \in [g]$. 

The \emph{asymmetric quantum approximate cloning} problem asks whether a quantum channel exists which \emph{approximately} clones any input state. The degree of approximation can vary with the index of the marginal (i.e.~clone) in the asymmetric setting. Symmetric approximate cloning was completely described in \cite{werner1998optimal,keyl1999optimal} (using different figures of merit for the quality of the clones), while the asymmetric case was studied in \cite{studzinski2014group, kay2016optimal}. Physically, approximate cloning can be seen as a way to go around the obstruction from the no-cloning theorem by adding noise: our goal is to produce imperfect, noisy copies of the original input state. We formalize the above in the following definition (see also \cite{bluhm2018joint}). 
	
\begin{definition} 
The approximation parameters of physical $1 \to g$ asymmetric cloners on $\C^d$ are described by the following set:
\begin{align*}
    \Gamma^{clone}(g,d)&:=\Big\{\mathbf s\in[0,1]^g \, : \, \exists \Phi: \M_d \to \M_d^{\otimes g} \text{ quantum channel such that}
	\\
	&\forall \rho \in \M_d,\forall j \in [g],\quad \Tr_{[g] \setminus \{j\}}\Phi(\rho)=s_j \rho +(1-s_j)\frac{I_d}{d} \Big\}.
\end{align*}
\end{definition}	

The classical no-cloning theorem states that perfect clones are impossible: for all $g,d \geq 2$, $(1,1, \ldots, 1) \notin \Gamma^{clone}(g,d)$. In \cite{kay2016optimal},  the \emph{optimal asymmetric cloning parameters} were computed explicitly (see also \cite{studzinski2014group} for an alternative approach, based on representation theory). Those results, stated in term of fidelities, can be restated in our language of depolarizing channels using \cite[Proposition 6.5]{bluhm2020compatibility}, which uses the \emph{twirling} operation to symmetrize the marginals of an optimal cloner. 

\begin{theorem}\cite[Section 2.3, Theorem 1]{kay2016optimal}
For all $g,d\geq 2$, the \emph{optimal} asymmetric cloning parameters are given by
	\begin{align*}
	    \partial\Gamma^{clone}(g,d)&=\{\mathbf s\in (0,1]^g \, : \forall \epsilon > 0, \, (1+\epsilon)s \notin \Gamma^{clone}(g,d)\,\} \\
	    &= \Bigg\{\mathbf s\in(0,1]^g \, :\,  (g+d-1)\left[g-d^2+d+(d^2-1)\sum^g_{i=1} s_i\right]  = \\
	    &\qquad \qquad \qquad \qquad \left(\sum^g_{i=1}\sqrt{s_i(d^2-1)+1}\right)^2\Bigg\}.
	\end{align*}
	\end{theorem}
	
The task of cloning quantum states can be reinterpreted in the Heisenberg picture of quantum mechanics by looking at the dual map of a channel; this operation acts naturally on quantum measurements. In this picture, the dual property of producing imperfect clones is having noisy measurements. Let us define the asymmetric dual map for the POVMs, and the corresponding set of cloning parameters. Consider the set of parameters for this dual maps: 
\begin{align}
    \label{eq:def-tilde-Gamma-clone}\tilde{\Gamma}^{clone}(g,d)&:=\Big\{\mathbf s\in[0,1]^g \, : \, \exists \Psi: \M_d^{\otimes g}\to \M_d \text{ unital and completely positive such that}\\
	\nonumber&\forall X\in \M_d,\, \forall j \in [g],\quad \Psi(I^{\otimes (j-1)}\otimes X \otimes I^{\otimes (g-j)})=s_j X + (1-s_j)\frac{\Tr X }{d}I\Big\}.
\end{align}

\begin{proposition}
    The dual and the primal sets of cloning parameters are identical: $\forall g,d \geq 2$,   
    $$\tilde{\Gamma}^{clone}(g,d)=\Gamma^{clone}(g,d).$$
\end{proposition}
\begin{proof}
Let us prove the first inclusion $\tilde{\Gamma}^{clone}(g,d)\subseteq\Gamma^{clone}(g,d)$, the other one being similar. Let $\mathbf s\in\tilde{\Gamma}^{clone}(g,d)$, and consider the unital completely positive map $\Psi:\M_d^{\otimes g}\to \M_d$ having the tuple $\mathbf s$ as an approximation parameter. Let us define $\Phi:=\Psi^*$; since $\Psi$ is unital and completely positive, $\Phi$ is a quantum channel \cite[Section 2.2]{watrous2018theory}. For any quantum state $\rho \in \M_d^{1,+}$, any matrix $X \in \M_d$, and any $j \in [g]$, we have 
\begin{align*}
    \Tr \left[ \left(\Tr_{[g] \setminus \{j\}} \Phi(\rho) \right)\cdot X\right] &=  \Tr \left[ \Phi(\rho) \cdot \left(I_d^{\otimes(j-1)} \otimes X \otimes I_d^{\otimes (g-j)}\right)\right]\\
    &=\Tr \left[ \rho \cdot \Psi\left(I_d^{\otimes(j-1)} \otimes X \otimes I_d^{\otimes (g-j)}\right)\right]\\
    &=\Tr \left[ \rho \cdot \left(s_j X + (1-s_j) \frac{\Tr X}{d}I_d \right)\right]\\
    &= s_j \Tr[\rho X] + (1-s_j) \frac{\Tr X}{d}\\
    &= \Tr \left[ \left(s_j \rho + (1-s_j)\frac{I_d}{d} \right)\cdot X\right],
\end{align*}
proving that, for all $\rho$ and $j$, $\Tr_{[g] \setminus \{j\}} \Phi(\rho) = s_j \rho + (1-s_j)\frac{I_d}{d}$. Hence, $\Phi=\Psi^*$ is a valid quantum cloner with parameter $\mathbf s$, which finishes the proof.
\end{proof}
    
\medskip

We shall now use the above results on quantum cloning to generalize the following compatibility criterion. We denote by  $\lambda_{\min}(X)$ the minimal eigenvalue of a self-adjoint operator $X$.
    
\begin{proposition}{\cite[Proposition III.3]{heinosaari2020random}}\label{prop:clone}
	Consider two POVMs $A$ and $B$ on $\M_d$ satisfying
	\begin{align*}
	    \lambda_{\min}(A_i) &\geq \frac{1}{2(d+1)} \Tr A_i \qquad \forall i \\
	    \lambda_{\min}(B_j) &\geq \frac{1}{2(d+1)} \Tr B_j \qquad \forall j.
	\end{align*}
	Then, $A$ and $B$ are compatible. 
\end{proposition}

We provide next a generalization of the compatibility criterion above for $g$-tuples of POVMs and asymmetric noise parameters.

\begin{theorem}\label{thm:asym-cloning-criterion}
	Let $\mathbf A = (A^{(1)}, \ldots, A^{(g)})$ be a $g$-tuple of POVMs on $\M_d$ having, respectively, $k_1, \ldots, k_g$ outcomes. Define, for all $x \in [g]$, 
	$$s_x := 1 - \min_{i \in [k_x]} \frac{d\lambda_{\min}(A^{(x)}_i)}{\Tr A^{(x)}_i} \in [0,1].$$
	If  $\mathbf s \in \Gamma^{clone}(g,d)$, then the POVMs in $\mathbf A$ are compatible. 
\end{theorem} 
\begin{proof}
    Note first that the assumptions in the statement are equivalent to the following set of inequalities: 
    \begin{equation}\label{eq:lambda-min-A}
        \forall x \in [g], \, \forall i \in [k_x], \qquad \lambda_{\min}(A^{(x)}_i) \geq \frac{1-s_x}{d} \Tr A^{(x)}_i.    
    \end{equation}
    Let $\Psi$ be the unital completely positive map appearing in the definition of $\tilde \Gamma^{clone}(g,d) \ni \mathbf s$. Let us define, for all $x \in [g]$ such that $s_x >0$,
    $$B_i^{(x)}:=\frac{1}{s_x}\left( A_i^{(x)} - (1-s_x) \frac{\Tr A^{(x)}_i}{d}I_d \right), \quad \forall i \in [k_x].$$
    If $s_x = 0$, put $B^{(x)}_i = I_d / k_x$ for all $i \in [k_x]$. We claim that $\mathbf B = (B^{(x)})_{x \in [g]}$ form a tuple of POVMs on $\M_d$. Indeed, it is easy to see that $B^{(x)}$ is normalized for all $x$, and that the positivity of $B_i^{(x)}$ follows from Eq.~\eqref{eq:lambda-min-A} for all $i$. Moreover, we have $\Tr B^{(x)}_i = \Tr A^{(x)}_i$ for all $x,i$.
    
    Define, for $\mathbf i = (i_1, \ldots, i_g) \in [k_1] \times \cdots \times [k_g]$, 
    $$C_{\mathbf i} := \Psi(B^{(1)}_{i_1} \otimes \cdots \otimes B^{(g)}_{i_g}).$$
    Since $\Psi$ is (completely) positive and unital, it follows that $C$ is a POVM on $\M_d$ with $k_1 \cdots k_g$ outcomes. From \eqref{eq:def-tilde-Gamma-clone}, it follows that the $x$-marginal of $C$ is given by
    $$\forall i_x \in [k_x], \quad \sum_{i_1, \ldots, i_{x-1}, i_{x+1}, \ldots, i_g} \!\!\!\!\!\!\!\!\!\!\!\!\!\!\!C_{\mathbf i} = \Psi\left( I_d^{\otimes (x-1)} \otimes B^{(x)}_{i_x} \otimes  I_d^{\otimes (g-x)} \right) = s_x B^{(x)}_{i_x} + (1-s_x) \frac{\Tr B^{(x)}_{i_x}}{d}I_d = A^{(x)}_{i_x},$$
    showing that the POVMs $\mathbf A$ are compatible, with joint measurement $C$.
\end{proof}
	
Note that Proposition \ref{prop:clone} follows from Theorem \ref{thm:asym-cloning-criterion} using the fact that $$\left(\frac{d+2}{2(d+1)},\frac{d+2}{2(d+1)}\right) \in \Gamma^{clone}(2, d)$$
for all $d \geq 2$.
	
\section{Compatibility dimensions --- definition and examples}\label{sec:comp-dim}
	
This section contains the definition of the main objects we study in the paper: the different notions of \emph{compatibility dimension}. 
	
We start with an example in order to provide some intuition about dimension reduction. Consider $A = \{\ketbra{i}{i}\}_{i=1}^5$ the von Neumann measurement in the computational basis of $\C^5$, and the POVM $B = (B_i)_{i=1}^5$ given by
\begin{align*}
	B_1 = \frac 1 2 \begin{bmatrix}
		1 & 1 & 0 & 0 & 0\\
		1 & 1 & 0 & 0 & 0\\
		0 & 0 & 0 & 0 & 0\\
		0 & 0 & 0 & 0 & 0\\
		0 & 0 & 0 & 0 & 0
	\end{bmatrix}&, \qquad B_2 = \frac 1 2 \begin{bmatrix}
		1 & -1 & 0 & 0 & 0\\
		-1 & 1 & 0 & 0 & 0\\
		0 & 0 & 0 & 0 & 0\\
		0 & 0 & 0 & 0 & 0\\
		0 & 0 & 0 & 0 & 0
	\end{bmatrix}, \\
	B_3 = \frac 1 2 \begin{bmatrix}
		0 & 0 & 0 & 0 & 0\\
		0 & 0 & 0 & 0 & 0\\
		0 & 0 & 1 & 1 & 0\\
		0 & 0 & 1 & 1 & 0\\
		0 & 0 & 0 & 0 & 0
	\end{bmatrix}, \qquad B_4 = \frac 1 2 &\begin{bmatrix}
		0 & 0 & 0 & 0 & 0\\
		0 & 0 & 0 & 0 & 0\\
		0 & 0 & 1 & -1 & 0\\
		0 & 0 & -1 & 1 & 0\\
		0 & 0 & 0 & 0 & 0
	\end{bmatrix}, 
	\qquad B_5 = \begin{bmatrix}
		0 & 0 & 0 & 0 & 0\\
		0 & 0 & 0 & 0 & 0\\
		0 & 0 & 0 & 0 & 0\\
		0 & 0 & 0 & 0 & 0\\
		0 & 0 & 0 & 0 & 1
	\end{bmatrix}.
\end{align*}
Note that we have $A_5=B_5 = \ketbra{5}{5}$. On the two-dimensional space spanned by $\ket 1$, $\ket 2$ (resp.~$\ket 3$, $\ket 4$), the operators $A_{1,2}$ and $B_{1,2}$ (resp.~$A_{3,4}$ and $B_{3,4}$) perform the von Neumann measurements in the two bases below (left basis for $A$ and right basis for $B$): 
\begin{center}
    \includegraphics{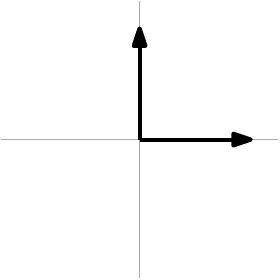} \qquad\qquad\qquad\qquad\qquad \includegraphics{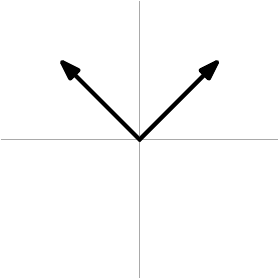}
\end{center}
Since the projective measurements $A,B$ do not correspond to the same orthonormal basis, they are not compatible. However, one can render them compatible by considering their \emph{reduction} (see Definition \ref{def:reduced-POVMs}) on a three-dimensional space. Indeed, consider the isometry $V :\C^3 \to \C^5$ given by 
\begin{equation}\label{eq:def-V}
    V = \ketbra{1}{1} + \ketbra{3}{2} + \ketbra{5}{3}.
\end{equation}
We have $$V^*AV = (\ketbra{1}{1}, 0, \ketbra{2}{2}, 0, \ketbra{3}{3})$$ while
$$V^*BV = \left(\frac{\ketbra{1}{1}}{2}, \frac{\ketbra{1}{1}}{2}, \frac{\ketbra{2}{2}}{2}, \frac{\ketbra{2}{2}}{2}, \ketbra{3}{3}\right).$$
Hence, although the original POVMs $A$, $B$ were incompatible, their reduced versions $V^*AV$ and $V^*BV$ are commuting, hence compatible. From a physical perspective, we have found a  $3$-dimensional subspace $E = \operatorname{Ran}(V)\subseteq \C^5$ such that the POVMs $A,B$ look compatible when measuring quantum states supported on $E$. This connection with quantum states shall be discussed in details in Section \ref{sec:IW}.

\medskip

We now introduce the main quantities of interest in this work, starting with the most general one. We recall that, in the theory of partial ordered sets, a \emph{down-set} is a set $X$ with the property that if $x \in X$ and $y \preceq x$, then $y \in X$ (``$\preceq$'' denotes the partial order relation).

\begin{definition}\label{def:C-A}
	Given a $g$-tuple of POVMs $\mathbf A=(A^{(1)}, \ldots, A^{(g)})$, define their \emph{compatibility down-set} as 
	\begin{equation}
	    \mathcal C(\mathbf A) := \{E \subseteq \mathbb C^d \, | \, V^* \mathbf A V \text{ are compatible for some isometry $V$ with $\operatorname{Ran}(V)=E$}\}.
	\end{equation}
	In other words, the compatibility down-set is the set of subspaces on which the POVMs $\mathbf A$ are compatible. 
\end{definition}

We gather some basic facts about the sets $\mathcal C(\mathbf A)$ in the following proposition. We denote by  $\mathcal S_r(\mathbb C^d)$ the Grassmannian of all $r$-dimensional subspaces of $\mathbb C^d$
$$\mathcal S_r(\mathbb C^d) := \{E \subseteq \mathbb C^d\, | \, \dim E = r \}$$
and we also write 
$$\mathcal S(\mathbb C^d) = \bigsqcup_{r=0}^d \mathcal S_r(\mathbb C^d)$$
for the full Grassmannian. 

\begin{proposition}
    The set $\mathcal C(\mathbf A)$ has the following properties: 
    \begin{itemize}
        \item $\mathcal C(\mathbf A)$ is a down-set in the modular lattice $\mathcal S(\mathbb C^d)$ of subspaces of $\mathbb C^d$ 
        \item $\mathcal C(\mathbf A)$ contains all the 1-dimensional subspaces
        \item the POVMs $\mathbf A$ are compatible if and only if $\mathcal C(\mathbf A) = \mathcal S(\mathbb C^d)$
        \item $\mathcal C(\mathbf A)$ is graded by $r=\dim E$: 
    $$\mathcal C(\mathbf A) = \bigsqcup_{r=0}^d \mathcal C_r(\mathbf A),$$
    where
    $$\mathcal C_r(\mathbf A) := \mathcal C(\mathbf A) \cap \mathcal S_r(\mathbb C^d).$$
    \item in Definition \ref{def:C-A}, the words ``some isometry'' can be replaced by ``all isometries''. 
    \end{itemize}
\end{proposition}
\begin{proof}
Let us prove the first claim. Consider a subspace $F \subseteq E$ of dimension $\dim F = s$ and choose an isometry $W: \mathbb C^s \to \mathbb C^d$ such that $\operatorname{Ran} W= F$. Since $F \subseteq E$, we have $W = VV^*W$. We have thus $W^* \mathbf{A} W = W^* V (V^* \mathbf A V )V^*W$. The compatibility of $W^* \mathbf{A} W$ follows then from that of $V^* \mathbf{A} V$.

The fact that $\mathcal C(\mathbf A)$ contains all vector lines follows from commutativity. Having $\mathbb C^d \in \mathcal C(\mathbf A)$ is clearly equivalent to the compatibility of the POVMs in $\mathbf A$. 

The final claim follows from the observation that any two isometries $V_{1,2}:\mathbb C^r \to \mathbb C^d$ with $\operatorname{Ran} V_{1,2}= E$ are related via a unitary $U : \mathbb C^r \to \mathbb C^r$ by $V_2 = V_1 U$, and from the fact that conjugation by a global unitary does not change compatibility. 
\end{proof}

\begin{remark}
    The map $\mathbf A \mapsto \mathcal C(\mathbf A)$ is an anti-order-morphism with respect to the pre- and post-processing order relations on the set of tuples of POVMs, see \cite[Section 5]{heinosaari2016invitation}.
\end{remark}

Since the lattice of subspaces of $\mathbb C^d$ is a cumbersome object to work with, we consider a coarse-grained version of Definition \ref{def:C-A}, where we keep track only of the dimension of the subspaces. 

\begin{definition}\label{def:R-Rbar}
Given a $g$-tuple of POVMs $\mathbf A=(A^{(1)}, \ldots, A^{(g)})$ on a $d$-dimensional quantum system, we define their \emph{compatibility dimension} as the largest dimension $r$ for which there \emph{exists} an isometry $V: \C^r \to \C^d$ reducing the POVMs to a compatible $g$-tuple: 
	\begin{align}
		R(\mathbf A) &:= \max\{ r \in [d] \, : \, \exists V : \C^r \to \C^d \text{ isom.~s.t. } V^*A^{(1)}V, \ldots, V^*A^{(g)}V \text{ are comp.}\}\\
		\nonumber &= \max\{ r \in [d] \, : \mathcal C_r(\mathbf A) \neq \emptyset\}.
	\end{align}
	
	Similarly, we define the \emph{strong compatibility dimension} of a $g$-tuple of POVMs $\mathbf A$ as the largest dimension $r$ for which \emph{all} isometries $V: \C^r \to \C^d$ reduce the POVMs to a compatible $g$-tuple:
	\begin{align}
		\bar R(\mathbf A) &:= \max\{ r \in [d] \, : \, \forall V : \C^r \to \C^d \text{ isom., } V^*A^{(1)}V, \ldots, V^*A^{(g)}V \text{ are comp.}\}\\
		\nonumber &= \max\{ r \in [d] \, : \mathcal C_r(\mathbf A)  = \mathcal S_r(\mathbb C^d)\}.
	\end{align}
\end{definition}

We have the following simple observations, which follow directly from the definition. 
\begin{remark}
	For all $g$-tuples $\mathbf A$ of POVMs on $\M_d$, we have 
	$$1 \leq \bar R(\mathbf A) \leq R(\mathbf A) \leq d.$$
	We also have $\bar R(\mathbf A) = d \iff R(A) = d \iff A^{(1)}, \ldots, A^{(g)}$  are compatible quantum measurements.
\end{remark}

For the example of the two POVMs $A,B$ introduced at the beginning of this section, using the isometry $V$ from \eqref{eq:def-V}, we have $R(A,B) \geq 3$. On the other hand, using the isometry 
$$W  = \ketbra{1}{1} + \ketbra{2}{2} + \ketbra{5}{3},$$
we have $W^*AW = (\ketbra{1}{1}, \ketbra{2}{2}, 0, 0, \ketbra{3}{3})$, while
$$W^*BW = \left( \frac 1 2 \begin{bmatrix}
		1 & 1 & 0 \\
		1 & 1 & 0 \\
		0 & 0 & 0
	\end{bmatrix}, \frac 1 2 \begin{bmatrix}
		1 & -1 & 0 \\
		-1 & 1 & 0 \\
		0 & 0 & 0 
	\end{bmatrix}, 0, 0, \begin{bmatrix}
		0 & 0 & 0\\
		0 & 0 & 0\\
        0 & 0 & 1
	\end{bmatrix}\right).$$
Note that the two POVMs $W^*AW, W^*BW$ are incompatible, proving that $\bar R(A,B) \leq 2$; we have thus provided an example where $\bar R < R$.

In this work, we shall focus mostly on the quantity $R$. Let us point out however that the measure $\bar R$ has been related in \cite{bluhm2018joint,bluhm2020compatibility} to the inclusion problem for different levels of the matrix diamond and its generalizations into a free spectrahedron defined by $\mathbf A$; we shall not pursue these aspects in this work.  

\section{Restricted incompatibility witnesses}\label{sec:IW}

We provide in this section a characterization of the incompatibility dimension with the help of incompatibility witnesses. This point of view is ``dual'' in some sense to the original definition from Section \ref{sec:comp-dim}, providing an operational interpretation of the dimensions $R(\mathbf{A})$ and $\bar R(\mathbf{A})$ as the size of the support of superensembles of quantum states allowing for an advantage in a state discrimination protocol (see Theorem \ref{thm:R-bar-R-discrimination}).

Several notions of incompatibility witnesses have been considered in the literature, by \cite{jencova2018incompatible}, \cite{carmeli2019quantum}, and \cite{bluhm2020compatibility}. We shall consider here the second listed approach, developed in \cite{carmeli2018state, carmeli2019quantum}, which has a very nice operational interpretation, in terms of state ensembles distinguishability, with prior vs.~posterior information. The same connection between incompatibility witnesses and state ensemble distinguishability was discovered independently in \cite{oszmaniec2019operational,skrzypczyk2019all,uola2019quantifying}.

Let us first describe the state discrimination protocols which provide the framework for incompatibility witnesses, following \cite{carmeli2018state}. Recall that a state ensemble $\mathcal E$ is a set of quantum states $\sigma_1, \ldots, \sigma_k \in \mathcal M_d^{1,+}(\mathbb C)$, together with a probability vector $p=(p_1, \ldots, p_k)$. We also consider \emph{superensembles} $\mathbfcal E$, which are $g$-tuples of state ensembles $(\mathcal E^{(1)}, \ldots, \mathcal E^{(g)})$, together with a probability measure $q=(q_1, \ldots, q_g)$. Note that we do not require that the number of elements in each ensemble (respectively $k_1, \ldots, k_g$) is identical. We consider now two superensemble discrimination protocols, which differ only in the timing when the state ensemble label is communicated. The main idea of the protocol is presented in Figure \ref{fig:superensemble-discrimination}, while the details of the explicit steps of the protocol are given in Table \ref{tab:superensemble-discrimination}. 

\begin{figure}[htb!]
    \centering
    \includegraphics{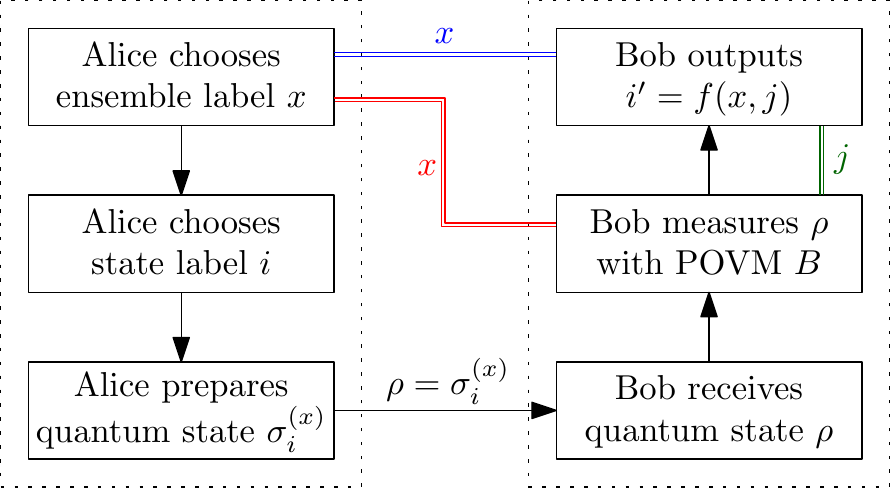}
    \caption{The superensemble discrimination protocol, with its two variants: \textcolor{red}{prior information} and \textcolor{blue}{posterior information}.}
    \label{fig:superensemble-discrimination}
\end{figure}

The input of the protocol is a superensemble $\mathbfcal{E}$, and we shall be interested in the success probability $\mathbb P_{guess}$, of Bob correctly identifying to which ensemble element Alice's state corresponds to. In other words, we are interested in Bob's best choice of a POVM $B$ such that the probability that the protocol succeeds (i.e.~$i=i'$) is maximal. Let us consider the two scenarios separately. In the scenario with prior information, Bob knows from which ensemble $\mathcal E^{(x)}$ the state $\rho$ has been sampled, so he can choose $B$ to be the POVM which discriminates best the (weighted) states from $\mathcal E^{(x)}$. We obtain
$$\mathbb P_{guess}^{\, prior}(\mathbfcal E) = \sup \left\{ \sum_{x=1}^g q_x \langle \mathcal E^{(x)}, B^{(x)} \rangle \, : \, B^{(1)}, \ldots, B^{(g)}\text{ POVMs}\right\},$$
where we use $\langle \cdot, \cdot \rangle$ to denote the state ensemble-POVM duality:
$$\langle \mathcal E^{(x)}, B^{(x)} \rangle := \sum_{i=1}^{k_x} p^{(x)}_i \Tr[\sigma^{(x)}_i B^{(x)}_i].$$

In the scenario with posterior information, Bob does not have the knowledge of $x$ at the time he performs the quantum measurement, and it has been shown in \cite[Eq.~(13)]{carmeli2018state} that 
$$\mathbb P_{guess}^{\, post}(\mathbfcal E) = \sup \left\{ \sum_{x=1}^g q_x \langle \mathcal E^{(x)}, C^{(x)} \rangle \, : \, C^{(1)}, \ldots, C^{(g)}\text{ compatible POVMs}\right\},$$
The formula above can be understood as follows: since at the time he performs the measurement, Bob does not know from which ensemble $\mathcal E^{(x)}$ the state $\rho$ is sampled from, his best bet is to perform a measurement with a large outcome set and then, once he learns the ensemble label $x$, to perform a classical post-processing of his measurement outcome $j$ and the ensemble label $x$. This classical post-processing is equivalent to Bob measuring a joint POVM $C$ of \emph{compatible} POVMs $C^{(1)}, \ldots, C^{(g)}$, having respectively $k_1, \ldots, k_g$ outcomes, see \cite[Proposition 1]{carmeli2018state}. Since the set over which the supremum is considered is smaller in this scenario, we have $\mathbb P_{guess}^{\, prior}(\mathbfcal E) \geq \mathbb P_{guess}^{\, post}(\mathbfcal E)$.

\renewcommand{\arraystretch}{2}	
\begin{table}[ht]
\begin{tabular}{|r|c|l|}
\hline
\cellcolor[HTML]{C0C0C0} \textbf{Step} & \cellcolor[HTML]{C0C0C0}\textbf{\textcolor{red}{Prior information}} & \multicolumn{1}{c|}{\cellcolor[HTML]{C0C0C0}\textbf{\textcolor{blue}{Posterior information}}} \\ \hline
1    & \multicolumn{2}{c|}{Alice chooses randomly an ensemble label $x \in [g]$, using probabilities $q$}                               \\ \hline
2    & \multicolumn{2}{c|}{Alice chooses randomly a state label $i \in [k_x]$, using probabilities $p^{(x)}$}                           \\ \hline
3    & \multicolumn{2}{c|}{Alice sends the quantum state $\rho = \sigma^{(x)}_i$ to Bob}                                                \\ \hline
4    & \textcolor{red}{Alice sends the ensemble label $x$ to Bob}          &                                                                             \\ \hline
5    & \multicolumn{2}{c|}{Bob receives the (unknown) quantum state $\rho$}                                                              \\ \hline
6    & \multicolumn{2}{c|}{Bob chooses a POVM $B$ and measures $\rho$, obtaining an output $j$}                                         \\ \hline
7    &                                                    & \textcolor{blue}{Alice sends the ensemble label $x$ to Bob}                                   \\ \hline
8    & \multicolumn{2}{c|}{Bob outputs $i' = f(x,j)$}                                                                                   \\ \hline
9    & \multicolumn{2}{c|}{The protocol succeeds if $i'=i$}                                                                             \\ \hline
\end{tabular}
\caption{Superensemble discrimination protocols, with prior and posterior information. In the \textcolor{red}{prior information} scenario, Alice sends Bob the ensemble label $x$ \textcolor{red}{before} Bob makes his measurement, allowing him to choose a POVM depending on the value $x$. In the \textcolor{blue}{posterior information} scenario, Bob only learns $x$ \textcolor{blue}{after} performing his measurement, which cannot depend on $x$.}
\label{tab:superensemble-discrimination}
\end{table}
\renewcommand{\arraystretch}{1}

Next, Carmeli, Heinosaari and Toigo define incompatibility witnesses as follows. 
\begin{definition}[\cite{carmeli2018state,carmeli2019quantum}]
    An \emph{incompatibility witness} is a superensemble $\mathbfcal E$ such that 
    $\mathbb P_{guess}^{\, prior}(\mathbfcal E) > \mathbb P_{guess}^{\, post}(\mathbfcal E)$.
\end{definition}

Incompatibility witnesses are used to detect incompatibility of $g$-tuples of POVMs in an obvious manner: given $\mathbf A = (A^{(1)}, \ldots A^{(g)})$, we have
\begin{equation}\label{eq:IW-implies-incomp}
\sum_{x=1}^g q_x \langle \mathcal E^{(x)}, A^{(x)} \rangle =:\langle \mathbfcal E, \mathbf A \rangle > \mathbb P_{guess}^{\, post}(\mathbfcal E) \implies \mathbf A \text{ are incompatible}.
\end{equation}

Obviously, for any $g$-tuple of POVMs $\mathbf A$, we have $\langle \mathbfcal E, \mathbf A \rangle \leq \mathbb P_{guess}^{\, prior}(\mathbfcal E)$; the incompatibility witness $\mathbfcal E$ detect the incompatibility of $\mathbf A$ only when 
$$\langle \mathbfcal E, \mathbf A \rangle \in ( \mathbb P_{guess}^{\, post}(\mathbfcal E),  \mathbb P_{guess}^{\, prior}(\mathbfcal E)].$$

Importantly, Carmeli, Heinosaari and Toigo establish the following converse to \eqref{eq:IW-implies-incomp}. 

\begin{theorem}{\cite[Theorem 2]{carmeli2019quantum}}
A $g$-tuple $\mathbf A$ of POVMs on $\M_d$ are compatible if and only if, for all incompatibility witnesses $\mathbfcal E$ on $\C^d$, we have 
$$\langle \mathbfcal E, \mathbf A \rangle \leq \mathbb P_{guess}^{\, post}(\mathbfcal E).$$
\end{theorem}
	
\medskip

We discuss now the relation between a restricted notion of incompatibility witnesses and the compatibility dimension we introduced in Section \ref{sec:comp-dim}. We start with the following important definition. 

\begin{definition}
    Given a subspace $H \subseteq \C^d$, we say that a quantum state $\sigma$ is \emph{supported} on $H$ if 
    $\operatorname{Ran}(\sigma) \subseteq H$. Equivalently, $\sigma$ is supported on $H$ if  $P_H \sigma P_H = \sigma$, where $P_H$ is the orthogonal projection on $H$. We say that an ensemble of quantum states $\mathcal E$ (resp.~a superensemble $\mathbfcal E$) is supported on $H$ if all the states $\sigma_i \in \mathcal E$ with $p_i > 0$ are supported on $H$. We define the corresponding notion for superensembles in a similar manner. 
\end{definition}

Our starting point is the following observation. Given an ensemble of quantum states supported on a subspace $H$ and a POVM $A$, we have, for an isometry $V : \C^{\dim H} \to \C^d$ with $\operatorname{Ran} V = H$:
\begin{align*}
    \langle \mathcal E, A \rangle &= \sum_{i=1}^k p_i \Tr[\sigma_i A_i] = \sum_{i=1}^k p_i \Tr[P_H\sigma_i P_H A_i] \\
    &= \sum_{i=1}^k p_i \Tr[VV^*\sigma_i VV^* A_i] = \sum_{i=1}^k p_i \Tr[V^*\sigma_iVV^* A_iV] = \langle V^*\mathcal EV, V^*AV \rangle,
\end{align*}
\\
On the other hand, any (compatible) g-tuple of POVMs $\mathbf B=(B^{(1)},\ldots,B^{(g)})$ on $\M_{dim H}$ can be written as $\mathbf B=V^*\mathbf A V$ where $\mathbf A=(A^{(1)},\ldots,A^{(g)})$ is a (compatible) g-tuple of POVMs on $\M_d$. Indeed it is enough to define $A_{i}^{(x)}=VB_{i}^{(x)}V^*+\frac{I_d-VV^*}{k_k}$ for all $i\in [k_x]$, where $k_x$ is the number of outcomes of $B^{(x)}$. This fact, together with the previous equation, immediately yields $\mathbb P_{guess}^{prior}(\mathbfcal E)=\mathbb P_{guess}^{prior}(V^*\mathbfcal E V)$ and $\mathbb P_{guess}^{post}(\mathbfcal E)=\mathbb P_{guess}^{post}(V^*\mathbfcal E V)$ for all superesnsembles $\mathbfcal E$ supported on $H$.
\\

We have the following result, relating (super)ensembles supported on subspaces to the (strong) compatibility dimension of POVMs. 

\begin{theorem}\label{thm:R-bar-R-discrimination}
    Given a $g$-tuple $\mathbf A$ of POVMs on $\M_d$ and an integer $r \in [d]$, we have $R(\mathbf A) \geq r$ if and only if there exists a subspace $H \in \mathcal S_r(\C^d)$ (i.e.~$H \subseteq \C^d$ with $\dim H = r$) such that for all superensembles $\mathbfcal E$ supported on $H$ we have 
    $$\langle \mathbfcal E, \mathbf A \rangle \leq \mathbb P_{guess}^{\, post}(\mathbfcal E).$$
    Similarly, $\bar R(\mathbf A) \geq r$ if and only if for all superensembles $\mathbfcal E$ supported on subspaces of dimension $r$, the relation above holds. 
\end{theorem}
\begin{proof}
We shall only prove the first claim, leaving the proof of the second claim to the reader. The condition $R(\mathbf A) \geq r$ is equivalent to the existence of an isometry $V: \C^r \to \C^d$ such that the POVMs $V^* \mathbf A V$ are compatible. Let us fix such an isometry $V:  \C^r \to \C^d$ with $\operatorname{Ran} V = H$ and start with the proof of the $\implies$ implication. For a superensemble $\mathbfcal E$ supported on $H$, we have 
$$\langle \mathbfcal E, \mathbf A \rangle =  \langle V^*\mathbfcal E V, V^*\mathbf A V \rangle  \leq \mathbb P_{guess}^{\, post}(V^*\mathbfcal E V) = \mathbb P_{guess}^{\, post}(\mathbfcal E),$$
proving the claim. The reverse implication follows the same reasoning: the equation above is still true, and all superensembles $\mathbfcal E'$ on $\C^r$ can be written as $V^* \mathbfcal E V$ for some $\mathbfcal E$ supported on $H$, namely $\mathbfcal E=V \mathbfcal E' V^*$ .
\end{proof}

To summarize, we have shown in this section that the compatibility dimensions of a $g$-tuple of POVMs can be understood in terms of a  superensemble distinguishability protocol, with states having restricted support in $\mathbb C^d$. 

\section{Two orthonormal bases}\label{sec:2-bases}

We consider in this section the case of two von Neumann measurements $A$ and $B$ corresponding to orthonormal bases in $\C^d$, say $\{\ket{a_i}\}_{i=1}^d$ and $\{\ket{b_i}\}_{i=1}^d$. The first observation that we can make is that we can assume, by a global unitary rotation, that one of the bases, say the first one, is the computational (canonical) basis in $\C^d$: $\ket{a_i} = \ket i$ for all $1 \leq i \leq d$. Let $U$ be the unitary operator implementing the change of basis, such that the second basis is given by the columns of $U$, $\{\ket{u_i}\}_{i=1}^d$. With this notation, our task is now to compute, for some given unitary matrix $U \in \mathcal U_d$,
$$\mathcal Z(U) := R\left( \{\ketbra{i}{i}\}_{i=1}^d,\{\ketbra{u_i}{u_i}\}_{i=1}^d \right).$$

Consider now an isometry $V:\C^r \to \C^d$ and note that the operators $\tilde A_i = V^*\ketbra{i}{i}V$ and $\tilde B_i = V^*\ketbra{u_i}{u_i}V$ have rank at most one. Compatibility of unit rank POVMs is essentially the same as equality, up to permutation of effect operators and summing together collinear effects \cite{kuramochi2015minimal,heinosaari2019post}. We have thus the following lower bound; we conjecture that the bound is tight for generic, non-degenerate unitary matrices. 

\begin{proposition}
For any unitary operator $U \in \mathcal U_d$, we have 
\begin{equation}\label{eq:bound-R-U}\mathcal Z(U) \geq \max_{\substack{z \in \C^d \\ \sigma \in \mathfrak S_d}} \dim \ker(P_{z,\sigma} -U),
\end{equation}
where $\mathfrak S_d$ is the symmetric group on $d$ elements, and $P_{z, \sigma}$ is the generalized permutation matrix given by
$$P_{z,\sigma}(i,j) = z_j \delta_{i,\sigma(j)}, \qquad \forall i,j\in[d].$$
\end{proposition}
\begin{proof}
First, note that in \eqref{eq:bound-R-U} one can consider the adjoint of the operator $P_{z,\sigma} -U$, since for any matrix $X \in \mathcal M_d$, we have $\dim \ker X = \dim \ker(X^*)$.  Consider a vector of scalars $z \in \C^d$ and a permutation $\sigma \in \mathfrak S_d$, and let $E = \ker\left[(P_{z,\sigma} -U)^*\right]$ having dimension $r:=\dim E$. We have then, for some isometry $V : \C^{r} \to \C^d$ with range $E$, 
$$(P_{z,\sigma} -U)^*V = 0_{d \times r} \implies V^*(P_{z,\sigma} -U) = 0_{r \times d}.$$
Hence, for any $j \in [d]$, we have 
$$V^*\ket{u_j} = z_jV^*\ket{\sigma(j)} \implies V^*\ketbra{u_j}{u_j}V = |z_j|^2 V^*\ketbra{\sigma(j)}{\sigma(j)}V.$$
Hence, $V^*AV$ and $V^*BV$ are compatible POVMs, having collinear effect operators. 
\end{proof}

We leave the question of computing $\mathcal Z(U)$ open in the general case. Even the bound from Eq.~\eqref{eq:bound-R-U} seems to be hard to compute in general. A trivial lower bound is given by the largest multiplicity of the eigenvalues of $U$, corresponding to taking a constant vector $z$ and fixing $\sigma = \mathrm{id}$. A natural candidate for the vector $z$ is the diagonal of $U$, i.e.~$z_i = u_{ii}$, a choice which has the merit that the matrices $P_{z,\mathrm{id}}$ and $U$ have identical diagonals. Imposing the additional constraint $|z_i| = 1$ (i.e.~$P_{z,\mathrm{id}}$ is unitary) amounts to choosing $z_i = \operatorname{phase}(u_{ii}) = u_{ii}/|u_{ii}|$, in the case of non-zero $u_{ii}$. These values are the solution of the following optimization problem:
$$\operatorname{argmin}_{z \in \C^d} \|P_{z,\mathrm{id}}-U\|_2^2,$$
with or without the additional constraint that $P_{z,\mathrm{id}}$ is unitary. The problem above is similar in nature to the bound from \eqref{eq:bound-R-U}: the objective functions correspond to the matrices $P_{z,\mathrm{id}}$ and $U$ being close to each other. 

\begin{example}\label{ex:Fourier}
    In the case of the Fourier operator $U = F_d$ given by $F_d(\alpha,\beta) = \omega^{\alpha\beta}$ with $\omega = \exp(2\pi\mathrm{i}/d)$, we have, with the choice $z_i = 1$ and $\sigma = \mathrm{id}$, 
    $$\mathcal Z(F_d) \geq 1+\lfloor d/4 \rfloor,$$
    using the eigenvalue $\lambda=1$ of $F_d$ \cite{mcclellan1972eigenvalue}. For example, in the case $d=4$, a basis of the 2-dimensional eigenspace associated to the eigenvalue $\lambda = 1$ is given by the following two vectors: 
    $$(1,0,1,0) \qquad \text{ and } \qquad (2,1,0,1).$$
    For the general case, the problem of constructing a ``simple'' eigenbasis of $F_d$ has received a lot of attention in the literature, see \cite{grunbaum1982eigenvectors,fendler2013discrete}.
\end{example}

\section{Complementary bases}
\label{sec:MUB}
We shall consider in this section the problem of dimension reduction for the special case of two (noisy) mutually unbiased bases. Recall that a set of $g$ orthonormal bases $\left\{\{\ket{b^{(x)}_i}\}_{i \in [d]} \right\}_{x \in [g]}$ are called \emph{mutually unbiased} (MUB) \cite{ivanovic1981geometrical,durt2010mutually} if $$\forall x \neq y \in [g], \, \forall i,j \in [d], \qquad |\braket{b^{(x)}_i | b^{(y)}_j}|^2=\frac{1}{d}.$$

Such kind of bases are very important in quantum information theory. For example, it was shown in \cite{wooters} that density matrices can be completely determined by making measurement in MUBs, and that this protocol is optimal, in the sense that the statistical error is minimized. The construction of such bases is deeply related to number theory and prime numbers which are very important for pure  mathematical investigation while they have several applications in quantum information theory, quantum cryptography and entanglement, tomography, etc.; see \cite{durt2010mutually}. 

Consider two mutually unbiased bases $\{\ket{a_1}, \ldots, \ket{a_d}\}$ and $\{\ket{b_1}, \ldots, \ket{b_d}\}$ in $\C^d$, for example the computational and the Fourier bases from Example \ref{ex:Fourier}. Let us introduce the noisy versions of the POVMs
\begin{align*}
	\mathcal N_\lambda[A] = \left( \lambda \ketbra{a_1}{a_1} + (1-\lambda)\frac{I_d}{d}, \ldots, \lambda \ketbra{a_d}{a_d} + (1-\lambda)\frac{I_d}{d}\right) \\
	\mathcal N_\mu[B] = \left( \mu \ketbra{b_1}{b_1} + (1-\mu)\frac{I_d}{d}, \ldots, \mu \ketbra{b_d}{b_d} + (1-\mu)\frac{I_d}{d}\right).
\end{align*}

The values $(\lambda, \mu)$ for which the POVMs above are compatible have been computed in \cite{carmeli2012informationally,carmeli2019quantum}: for $(\lambda, \mu) \in [0,1]^2$, $\mathcal N_\lambda[A]$ and $\mathcal N_\mu[B]$ are compatible iff
$$\lambda+\mu \leq 1 \text{ or } \lambda^2+\mu^2+\frac{2(d-2)}{d}(1-\lambda)(1-\mu) \leq 1.$$

We consider first the symmetric case $\lambda=\mu$. In this situation, the POVMs $\mathcal N_\lambda[A]$ and $\mathcal N_\lambda[B]$ are compatible if and only if 
\begin{equation}\label{eq:noisy-MUB-incompatible}
\lambda\leq\frac{1}{2}\left(1+\frac{1}{1+\sqrt{d}}\right).
\end{equation}
We shall show that for the same symmetric amount of noise and with a particular choice of an isometry $V : \C^r \to \C^d$, reducing the dimension of two incompatible noisy MUB measurements renders them compatible. 

\begin{theorem}\label{thm:reduce-MUBs}
    Consider two POVMs $A,B$ corresponding to a pair of mutually unbiased bases which can be extended to a triple of MUBs. For any  $2\leq r < \sqrt{d}$, there exists a non-empty interval $\Lambda_{r,d} \subset [0,1]$ (see Eq.~\eqref{eq:interval}) such that, for all $\lambda \in \Lambda_{r,d}$,
	\begin{itemize}
    \item the noisy MUB measurements $\mathcal N_{\lambda}[A]$, $\mathcal N_{\lambda}[B]$ are incompatible
    \item their reduced versions $V^* \mathcal N_{\lambda}[A] V$, $V^* \mathcal N_{\lambda}[B] V$ are compatible, 
\end{itemize}
    where $V: \mathbb C^r \to \mathbb C^d$ is an isometry obtained by truncating a third MUB.
\end{theorem}
Before giving the proof of the theorem, note that a triple of MUBs exists in every dimension, see \cite{klappenecker2003constructions,combescure2007mutually}.
\begin{proof}
    Consider a third basis $\{\ket{c_k}\}_{k=1}^d$ of $\mathbb C^d$ such that $\{a_i\}$, $\{b_j\}$, and $\{c_k\}$ form a set of three mutually unbiased bases. We define $V : \C^r \to \C^d$ as $V=\sum^{r}_{k=1}\ket{c_k}\bra{k}$; it is clear that $V$ is an isometry. 
	
	Note first that the range of parameters $\lambda$ for which the noisy POVMs $\mathcal N_{\lambda}[A]$, $\mathcal N_{\lambda}[B]$ are incompatible was computed in Eq.~\eqref{eq:noisy-MUB-incompatible}: 
	$$\frac{1}{2}\left(1+\frac{1}{1+\sqrt{d}}\right) < \lambda \leq 1.$$
	
	We shall now compute the range of the parameter $\lambda$ for which we can use Proposition \ref{prop:clone} in its symmetric version for the reduced POVMs $V^*\mathcal N_{\lambda}[A]V$ and $V^*\mathcal N_{\lambda}[B]V$ to certify their compatibility. Let us first calculate, for $i \in [d]$, $\lambda_{\min}(V^* \mathcal N_{\lambda}[A]_i V)$:
	
	$$\lambda_{\min}(V^* \mathcal N_{\lambda}[A]_i V)=\frac{1-\lambda}{d}+\lambda \cdot \lambda_{\min}\left[\sum_{k,l=1}^r\braket{c_k|a_i}\braket{a_i|c_l}\ketbra{k}{l}\right].$$
	Note that the operator in the bracket above has unit rank, hence the second term is null. We have thus $\lambda_{\min}(V^* \mathcal N_{\lambda}[A]_i V)=\frac{1-\lambda}{d}$, for all $i \in [d]$. A simple calculation gives $$\Tr V^* \mathcal N_{\lambda}[A]_i V=\frac{r}{d}.$$
	The same calculation can be performed, and the same result is obtained, for $V^*\mathcal N_{\lambda}[B]V$. 	Putting these together, we find that: 
	
	$$\lambda \leq \frac{2+r}{2(1+r)} \implies \begin{cases}
	\lambda_{\min}(V^*\mathcal N_{\lambda}[A]_iV) \geq \frac{1}{2(1+r)} \Tr V^*\mathcal N_{\lambda}[A]_i V &\qquad \forall i \in [d]\\
	\lambda_{\min}(V^*\mathcal N_{\lambda}[B]_j V) \geq \frac{1}{2(1+r)} \Tr V^*\mathcal N_{\lambda}[B]_j V&\qquad \forall j \in [d],
	\end{cases}$$
    showing that the assumptions of Proposition \ref{prop:clone} hold, and thus that the POVMs $V^*\mathcal N_{\lambda}[A]V$ and $V^*\mathcal N_{\lambda}[B]V$ are compatible for the respective range of $\lambda$. 
	
	Define now the interval
	\begin{equation}\label{eq:interval}
	    \Lambda_{r,d} := \left( \frac{2+\sqrt d}{2(1+\sqrt d)}, \frac{2+r}{2(1+r)}\right].
	\end{equation}
	From the computations above, we know that for all $\lambda \in \Lambda_{r,d}$, the POVMs satisfy the two points in the statement; the interval $\Lambda_{r,d}$ is non-empty as soon as $2\leq r < \sqrt d$.
\end{proof}

Let us now consider the asymmetric version of Theorem \ref{thm:reduce-MUBs}, where the amount on white noise added to each POVM can be different. We first introduce a generalization of the compatibility regions from \cite[Section III]{bluhm2018joint} and \cite[Definition 3.32]{bluhm2020compatibility}. 

\begin{definition}
Given a $g$-tuple $\mathbf A$ of $d$-dimensional POVMs, we define its \emph{restricted compatibility region} to be the subset 
	\begin{align*}[0,1]^g \ni \Delta(\mathbf A; r) = \{\mathbf s\in[0,1]^g \, : \,  &\exists V : \mathbb C^r \to \mathbb C^d \text{ s.t.~the reduced POVMs } V^* \mathcal N_{s_1}[A^{(1)}]V,\\
	& \quad V^* \mathcal N_{s_2}[A^{(2)}] V, \ldots,V^* \mathcal N_{s_g}[A^{(g)}] V \text{ are compatible} \}. 
	\end{align*}
\end{definition}

Using the generalization of the cloning criterion to asymmetric noise parameters from Theorem \ref{thm:asym-cloning-criterion}, we prove the following lower bound for the compatibility regions $\Delta(\mathbf A, r)$ for tuples of MUBs. 

\begin{proposition}
For any $g$-tuple of MUBs $\mathbf A$ which can be extended to a $(g+1)$-tuple of MUBs, we have 
	 $\Gamma^{clone}(g,r)\subseteq\Delta(\mathbf A;r)$.
\end{proposition}
\begin{proof}
	Let $\mathbf s\in\Gamma^{clone}(g,r)$, and consider the isometry $V:=\sum_{k=1}^r \ketbra{c_k}{k}$, where $\{\ket{c_k}\}_{k=1}^d$ is the $(g+1)$-th MUB from the statement. To conclude, it is enough to verify the assumptions of Theorem \ref{thm:asym-cloning-criterion}. The computations here are similar to the ones from Theorem \ref{thm:reduce-MUBs}. We have, for all $x \in [g]$ and $i \in [d]$,
	\begin{align*}
	    \lambda_{\min}(V^* \mathcal N_{s_x}[A^{(x)}]_i V) &= \frac{1-s_x}{d}\\
	     \Tr (V^* \mathcal N_{s_x}[A^{(x)}]_i V)&=\frac{r}{d}.
	\end{align*}
	Hence, 
	$$s_x = 1 - \min_{i \in [k_x]} \frac{d\lambda_{\min}(V^* \mathcal N_{s_x}[A^{(x)}]_i V)}{\Tr (V^* \mathcal N_{s_x}[A^{(x)}]_i V)}$$
	satisfies the hypothesis of Theorem \ref{thm:asym-cloning-criterion}.
\end{proof}

We leave the question of deriving upper bounds for the sets $\Delta(\mathbf{A}, r)$ open.

\section{Algebraic considerations}\label{sec:algebraic}

A simple way of using dimension reduction to render incompatible measurements compatible is to ensure that, after the reduction, the POVM elements of the measurements are commutative. Moreover, in the case of 2 POVMs, one can push this idea even further and render one of the reduced POVMs trivial, ensuring thus compatibility. The overarching theme of this section is to use the two algebraic characterizations of compatibility (commutativity and trivial POVMs) to obtain very general dimension reduction results. The price to pay for this generality is that, for some very specific situations, the results can be relatively weak, when compared with more specialized techniques, such as the ones from Sections \ref{sec:2-bases} and \ref{sec:MUB}. 

\bigskip
We start with a dimension reduction method by which POVMs are rendered commutative (and thus compatible). The following construction has been introduced in \cite[Theorem 3]{knill2000theory} and further refined in \cite[Proposition 2.4]{li2011generalized}. The connection with quantum error correction can be understood as follows: on the code space, the POVM channels act like the identity (up to a scalar), hence the reduced POVMs are trivial.  

For the sake of completeness, we recall it here in full details and adapt it to our setting, emphasizing the intermediate step related to commutative POVMs.  

\begin{definition}\label{def:commutativity-dimension}
	For a $g$-tuple of POVMs $\mathbf A=(A^{(1)}, \ldots, A^{(g)})$ on $\M_d$, we define their \emph{commutativity dimension} as
	\begin{align*}
		T(\mathbf A) &:= \max\{ r \in [d] \, : \, \exists V : \C^r \to \C^d \text{ isometry s.t. }\\
		\nonumber &\forall x \neq y \in [g], \, \forall i \in [k_x], \, \forall j \in [k_y],\quad  [V^*A^{(x)}_i V, V^*A^{(y)}_j V]=0\}. 
	\end{align*}
\end{definition}

We recall the following result from \cite{li2011generalized}, showing that tuples of matrices can be reduced to commutative operators, when the dimension is large enough. 
\begin{proposition}\label{prop:LB-C}
	Consider $m$ self-adjoint $d \times d$ matrices $A_1, \ldots, A_m$ and let 
	$$n+1 = \dim \operatorname{span}_{\mathbb R}\{A_1, \ldots, A_m, I_d\}.$$ 
	If $d \geq (n+1)(r-1)$, then there exist $r$ orthonormal vectors $x_1, \ldots, x_r \in \mathbb C^d$ such that, for all $s \in [m]$,  $\braket{x_i | A_s x_j} = 0$, whenever $i \neq j \in [r]$. In other words, the matrices $A$ are diagonal when restricted to the span of the vectors $\{x_1, x_2, \ldots, x_r\}$.
\end{proposition}
\begin{proof}
One can observe that if $\{B_1,\ldots,B_n\}$ is a basis of $\operatorname{span}_{\mathbb R}\{A_1, \ldots, A_m, I_d\}$, then there exist $r$ orthogonal vectors $x_1,\ldots,x_r$ such that $\braket{x_i|A_s x_j}=\lambda_s \delta_{ij}$ for all $i,j \in [r]$ and $s\in [m]$ iff the same holds true for the matrices $B_1,\ldots,B_n$. The result follows then from the first part of \cite[Proposition 2.4]{li2011generalized}.
\end{proof}

We shall now use the result above for the set of effects of a $g$-tuple of POVMs, to find an isometry reducing them to commuting POVMs. The following theorem combines Definition \ref{def:commutativity-dimension} with the lower bound from Proposition \ref{prop:LB-C}.

\begin{theorem}\label{thm:reduce-commutative}
	Consider a $g$-tuple $\mathbf A = (A^{(1)}, \ldots, A^{(g)})$, where $A^{(x)} = (A^{(x)}_1, \ldots , A^{(x)}_{k_x})$ is a POVM with $k_x$ outcomes. Let
	$$n+1 := \dim \operatorname{span}_{\mathbb R}\{A^{(x)}_i\}_{x\in[g], i \in [k_x]} \leq 1-g+\sum_{x=1}^g k_x.$$ 
	Then, we have the following lower bound:
	\begin{equation}\label{eq:LB-comm-POVM}
	R(\mathbf A) \geq T(\mathbf A) \geq 1+\left\lfloor \frac{d}{n+1} \right\rfloor \geq 1+\left\lfloor \frac{d}{1-g+\sum_{x=1}^g k_x} \right\rfloor.
	\end{equation}
\end{theorem}
\begin{proof}
For any $r \leq T(\mathbf A)$, there exists an isometry $V : \C^r \to \C^d$ such that the reduced effect operators $V^*A^{(x)}_iV \in \M_r$ commute with $V^*A^{(y)}_j V$ for all $i,j$ and $x\neq y$. In particular, the reduced POVMs $V^*A^{(x)}V$ are compatible: $R(\mathbf A) \geq T(\mathbf A)$. The second and third inequalities in \eqref{eq:LB-comm-POVM} follow from Proposition \ref{prop:LB-C}.
\end{proof}

\begin{remark}
	In the case where $n \geq d$, the lower bound \eqref{eq:LB-comm-POVM} is trivial. 
\end{remark}

\begin{remark}
	In the definition of $T(\mathbf A)$ we only ask that reduced effects from different POVMs commute, while the use of Proposition \ref{prop:LB-C} guarantees that all the reduced effects commute. It would be interesting to find out whether one can  gain something by exploiting this fact. 
\end{remark}

Let us illustrate the previous result by the following striking corollary, corresponding to the case $d=3$, $g=2$, $k_1 = k_2 = 2$.
\begin{corollary}
	Any pair of qutrit effects can be reduced to a pair of commuting (and thus compatible) qubit effects. 
\end{corollary}

\begin{example}\label{ex:qutrit-effects}
	Let us consider the following two qutrit effects, built from the computational and the Fourier bases in $\C^3$:
	$$E = \ketbra{1}{1} + \frac{\ketbra{2}{2}}{2} \qquad \qquad F = \ketbra{f_1}{f_1} + \frac{\ketbra{f_2}{f_2}}{2},$$
	where $f_{1,2,3}$ are the columns of the Fourier matrix 
	$$F_3 = \frac{1}{\sqrt 3}\begin{bmatrix}
	1 & 1 & 1 \\
	1 & \omega & \omega^2 \\
	1 & \omega^2 & \omega
	\end{bmatrix},$$
	with $\omega = \exp(2 \pi \mathrm{i}/3)$, see also Example \ref{ex:Fourier}. The fact that the effects $E,F$ are incompatible (that is, the POVMs $(E,I_3-E)$ and $(F, I_3-F)$ are incompatible) follows from the following semidefinite program \cite{boyd2004convex}:
	\begin{align*}
	\text{minimize} \quad &\lambda\\
	\text{subject to} \quad &X \geq 0\\
	&X \leq E\\
	&X \leq F\\
	&\lambda I_3 + X \geq E+F.
    \end{align*}
    In the SDP above, the variable $X$ corresponds to the single free value of a joint POVM for $E,F$. The effects $E,F$ are compatible if and only if the value of the SDP above is smaller or equal than one \cite[Eq.~(4)]{wolf2009measurements}. For our choice of $E,F$, it can be seen numerically that the value of the program is $\approx 1.577$, certifying the incompatibility of $E$ and $F$. 
    
 We choose the isometry
    $$V = \begin{bmatrix}
    1 & 0\\
    0 & \frac{1}{\sqrt{2}}\\
    0 & \frac{\omega}{\sqrt{2}}
    \end{bmatrix},$$
	for which the reduced effects read
	$$V^*EV = \begin{bmatrix}
    1 & 0\\
    0 & 1/4
    \end{bmatrix} \qquad   \text{ and } \qquad V^*FV=\begin{bmatrix}
    1/2 & 0\\
    0 & 1/2
    \end{bmatrix}.$$
    The reduced effects are commutative, hence compatible. 
\end{example}

\bigskip

We now move on to another method by which incompatible POVMs can be rendered compatible by dimension reduction. This time, we shall consider a single POVM and ``trivialize'' it by reducing it with an isometry. In the language of error correction, we are constructing a subspace of the Hilbert space on which the measurement channel acts like the identity. 

\begin{definition}
	Given a single POVM $A$ with $k$ outcomes on $\M_d$, its \emph{scalar dimension} is 
	$$S(A) := \max\{ r \in [d] \, : \, \exists V : \C^r \to \C^d \text{ isom.~s.t. } \forall i \in [k], \quad V^*A_iV \sim I_r\}.$$
\end{definition}

The definition above is related to the notion of \emph{higher rank (joint) numerical range} introduced in \cite{choi2006higher} for one matrix and generalized in \cite{li2011generalized} for several matrices. We recall the following lower bound from \cite[Proposition 2.4]{li2011generalized}, which uses Tverberg's theorem \cite{tverberg1966generalization} (see also \cite{barany2016tverbergs}) to render the diagonal matrices from Proposition \ref{prop:LB-C} multiples of the identity. 

\begin{proposition}\label{prop:LB-S}
	Consider $m$ self-adjoint $d \times d$ matrices $A_1, \ldots, A_m$ and let 
	$$n+1 = \dim \operatorname{span}_{\mathbb R}\{A_1, \ldots, A_m, I_d\}.$$ 
	If $d \geq (n+1)^2(r-1)$, then there exist $r$ orthonormal vectors $x_1, \ldots, x_r \in \mathbb C_d$ such that, for all $s \in [m]$, there exists a scalar $\lambda_s \in \R$ such that $\braket{x_i | A_s x_j} = \delta_{ij} \lambda_s$, for all $i,j \in [r]$.
\end{proposition}
\begin{proof}
The statement follows from the second part of the proof of \cite[Proposition 2.4]{li2011generalized}, by making the same observation as in the proof of Proposition \ref{prop:LB-C}.
\end{proof}

We can gather the results above in the following theorem.

\begin{theorem}\label{thm:reduce-trivial}
	Consider a pair of POVMs $A,B$ on $\M_d$. Let $k$ be the number of outcomes of the POVM $A$, and define
	$$n+1 := \dim \operatorname{span}_{\mathbb R}\{A_i\}_{i \in [k]} \leq k.$$ 
	We have the following lower bound:
	\begin{equation}\label{eq:LB-scalar-POVM}
	R(A,B) \geq S(A) \geq 1+\left\lfloor \frac{d}{(n+1)^2} \right\rfloor \geq 1+\left\lfloor \frac{d}{k^2} \right\rfloor.
	\end{equation}
\end{theorem}
\begin{proof}
For any $r \leq S(A)$ (resp.~$r \leq S(B)$), there exists an isometry $V : \C^r \to \C^d$ such that the POVM $V^*AV$ (resp.~$V^*BV$) is trivial. In particular, the POVMs $V^*AV$ and $V^*BV$ are compatible, and thus $R(A,B) \geq \max(S(A), S(B))$. The second inequality in \eqref{eq:LB-scalar-POVM} follows from Proposition \ref{prop:LB-S}.
\end{proof}
\begin{remark}
	In the case where $(n+1)^2 \geq d$, the lower bound \eqref{eq:LB-scalar-POVM} is trivial. In particular, if a POVM $A$ has $k$ linearly independent effects and $k > \sqrt d$, the bound \eqref{eq:LB-scalar-POVM} is trivial. Hence, Theorem \ref{thm:reduce-trivial} is useful for POVMs with few outcomes. 
\end{remark}

\begin{example}
	Going back to the two qubit effects from Example \ref{ex:qutrit-effects}, note that the reduced POVM $(V^*FV, I_2 - V^*FV)$ is the trivial POVM $(I_2/2,I_2/2)$.
\end{example}	

\medskip

To conclude, using ideas from the theory of quantum error correction, we have given in this section two lower bounds on the compatibility dimension of a tuple of POVMs $\mathbf A$: 
\begin{itemize}
    \item a first one in terms of the \emph{commutativity dimension} $T(\mathbf A)$ of the tuple, Theorem \ref{thm:reduce-commutative};
    \item a second one in terms of the \emph{scalar dimensions} $S(A)$ and $S(B)$ of any pair POVMs $(A,B)$, see Theorem \ref{thm:reduce-trivial}. 
\end{itemize}
We would like to point out that these very general results are useful in the regime where the POVMs have few outcomes (or, rather, the span of the effect operators is low-dimensional). The results in this section cannot be applied, for example, to the cases of (noisy) orthonormal bases that were studied in Sections \ref{sec:2-bases}, \ref{sec:MUB}. 
	
\section{Dimension dependent bounds and spin systems}\label{sec:spin-systems}

We prove in this section results for isometry-independent reductions, corresponding to the notion of strong compatibility dimension from Definition \ref{def:R-Rbar}.

We recall the following compatibility criterion from \cite[Section VIII]{bluhm2018joint} and \cite[Section 7]{bluhm2020compatibility} which guarantees the compatibility of noisy versions of POVMs, with a noise parameter depending on the dimension of the Hilbert space, and independent of the number of measurements. We shall explicitly consider separately the case of 2-outcome (or dichotomic) POVMs, with the example of maximally incompatible \emph{spin system measurements} in mind. 

\begin{proposition}{\cite[Corollary VIII.4]{bluhm2018joint} and \cite[Theorem 7.1]{bluhm2020compatibility}}\label{prop:comp-noise-dimension}
Let $A^{(1)}, \ldots, A^{(g)}$ be $g$ arbitrary 2-outcome POVMs on $\M_d$. Then, their noisy versions $\tilde A^{(x)}$ are compatible, where 
\begin{equation}\label{eq:dim-comp-crit-effects}
\tilde A^{(x)}_i = \mathcal N_{1/(2d)}[A^{(x)}]_i = \frac{1}{2d} A^{(x)}_i + \left(1-\frac{1}{2d}\right)\frac{I_d}{2}.
\end{equation}
More generally, consider a $g$-tuple $(B^{(x)})_{x=1}^g$, where $B^{(x)}$ is a $k_x$-valued POVM on $\mathbb C^d$. Then, their noisy versions $\tilde B^{(x)}$ are compatible, where 
\begin{equation}\label{eq:dim-comp-crit-general}
\tilde B^{(x)}_i = \mathcal N_{1/(2d(k_x-1)}[B^{(x)}]_i = \frac{1}{2d(k_x-1)} B^{(x)}_i + \left(1-\frac{1}{2d(k_x-1)}\right)\frac{I_d}{k_x}.
\end{equation}
\end{proposition}

This compatibility criterion is of particular interest in the setting of our work, given the dimension dependence of the noise parameters in the equations \eqref{eq:dim-comp-crit-effects} and \eqref{eq:dim-comp-crit-general}. Note that for small values of $g$, the compatibility result above can be seen to follow from other type of arguments, such as cloning \cite{heinosaari2014maximally}. We obtain the following universal lower bound on the quantity $\bar R(\cdot)$ from Definition \ref{def:R-Rbar}, giving thus the first lower bound on the strong compatibility dimension.
\begin{theorem}\label{thm:bound-R-bar-dimension-dependent}
Let $\mathbf A = (A^{(1)}, \ldots, A^{(g)})$ be a $g$-tuple of 2-outcome POVMs on $\M_d$. Then, for all $1 \leq r \leq d$ and $t \in [0, 1/(2r)]$, we have $\bar R(\mathcal N_t[\mathbf A]) \geq r$. 

More generally, consider a $g$-tuple $\mathbf B = (B^{(1)}, \ldots, B^{(g)})$, where $B^{(x)}$ is a $k_x$-valued POVM on $\mathbb C^d$. Then, for all $1 \leq r \leq d$ and $\mathbf t \in [0,1]^g$ such that $t_x \leq 1/(2r(k_x-1))$, we have $\bar R(\mathcal N_{\mathbf t}[\mathbf B]) \geq r$.
\end{theorem}
\begin{proof}
Let us prove the more general statement about the $g$-tuple $\mathbf B$. Fix an integer $r$ and a vector $\mathbf t$ as in the statement. Consider also an arbitrary isometry $V : \mathbb C^r \to \mathbb C^d$. From Lemma \ref{lem:noise-isometry}, we have that, for all $x \in [g]$,
$$V^*\mathcal N_{t_x} [B^{(x)}]V = \mathcal N_{t_x} [V^*B^{(x)}V].$$
Using Proposition \ref{prop:comp-noise-dimension} and the condition on the vector $\mathbf t$, we infer that the POVMs $\mathcal N_{\mathbf t}[V^*\mathbf B V]$ are compatible, proving the claim.
\end{proof}

\medskip

Let us now use the previous result to obtain bounds on the strong compatibility dimension of spin system measurements, which we introduce next. From a physical point of view \cite[Section 5.4]{weinberg}, it was discovered by Dirac that the spin property appears naturally in his equation when he was searching for a relativistic quantum equation of electrons. In his equation the Clifford algebra appears as a particular representation of the homogeneous Lorentz group. Since this representation contains naturally the spin one-half representation described by the Pauli matrices, his equation presents the conceptual and the natural description of the spin as a fundamental property. Mathematically, spin systems are sets of anti-commuting, self-adjoint, unitary operators. The paradigmatic example of such operators are the Pauli matrices $\sigma_{X,Y,Z} \in \mathcal M_2(\mathbb C)$. Higher level spin systems are defined recursively, as follows. At level $k=0$, we have a single matrix,  $$F^{(0)}_1 := [1] \in \mathcal M_1(\mathbb C).$$
At level $k = 1$, we have the Pauli matrices:
$$F^{(1)}_1 = \sigma_X = \begin{bmatrix}
0 & 1 \\ 1 & 0
\end{bmatrix}, \quad F^{(1)}_2 = \sigma_Y = \begin{bmatrix}
0 & -\text{i} \\ \text{i} & 0
\end{bmatrix} \quad \text{ and } \quad F^{(1)}_3 = \sigma_Z = \begin{bmatrix}
1 & 0 \\ 0 & -1
\end{bmatrix}.$$
For larger levels, define recursively the matrices of size $2^{k+1}$
$$F^{(k+1)}_i = \sigma_X \otimes F^{(k)}_i~\forall i \in [2k+1] \quad \text{ and } \quad F^{(k+1)}_{2k+2} = \sigma_Y \otimes I_{2^{k}}, \quad F^{(k+1)}_{2k+3} = \sigma_Z \otimes I_{2^{k}}.$$

For example, at level 2, we have the five matrices 
$$F^{(2)}_1 = \sigma_X \otimes \sigma_X, \qquad F^{(2)}_2 = \sigma_X \otimes \sigma_Y, \qquad F^{(2)}_3 = \sigma_X \otimes \sigma_Z, \qquad F^{(2)}_4 = \sigma_Y \otimes I_2, \qquad F^{(2)}_5 = \sigma_Z \otimes I_2.$$

From the $2k+1$ matrices at level $k$, we construct $2k+1$ dichotomic POVMs
$$A^{(x)}_1 = (I_{2^{k+1}}+F_x)/2 \qquad A^{(x)}_2 = (I_{2^{k+1}}-F_x)/2, \qquad x \in [2k+1].$$

We recall the following result from \cite{bluhm2018joint} regarding the noise robustness of the tuple $\mathbf A = (A^{(x)})_{x\in [2k+1]}$; note that the same result was derived in the symmetric case in \cite{kunjwal2014quantum}.

\begin{proposition}\cite[Section VIII.B]{bluhm2018joint}\label{prop:comp-spin-system}
For every $k \geq 1$, the $(2k+1)$-tuple of 2-outcome POVMs $\mathcal N_t[\mathbf A]$ acting on $\mathbb C^{2^{k+1}}$ is compatible if and only if $\|t\|_2 \leq 1$.
\end{proposition}

Combining the previous result with Theorem \ref{thm:bound-R-bar-dimension-dependent}, we obtain the following result, stating that, for appropriate noise parameters, the strong compatibility dimension of a noisy spin system POVM is neither 1 nor maximal. In other words, the noisy spin system POVMs are not compatible, but all reductions to a non-trivial fixed dimension become compatible. 

\begin{proposition}
For any $r \geq 2$, $k \geq 2r^2+1 \geq 9$, and all $t \in (1/\sqrt{2k+1},1/(2r)]$, the spin system POVMs $\mathbf A$ at level $k$ satisfy
$$ r \leq \bar R(\mathcal N_t[\mathbf A]) \leq 2^{k+1}-1.$$
\end{proposition}
\begin{proof}
The statement about compatibility follows from $t \leq 1/(2r)$ and Theorem \ref{thm:bound-R-bar-dimension-dependent}. The incompatibility statement follows from Proposition \ref{prop:comp-spin-system} and 
$$t > \frac{1}{\sqrt{2k+1}} \implies \|t(\underbrace{1,1,\ldots, 1}_{2k+1 \text{ times}})\|_2 >1.$$
The inequality between $k$ and $r$ ensures the existence of noise parameters for which the interval in the statement is non-empty. 
\end{proof}

\section{Conclusion}
In this paper, we have introduced a new measure of the incompatibility of a pair (or a tuple) of quantum measurements. The \emph{compatibility dimension} of a set of POVMs is the maximal dimension of a Hilbert space to which the restrictions of the given measurements are compatible. A related notion, that of the \emph{strong} compatibility dimension is defined in a similar manner, but requiring that the restrictions to  \emph{all} Hilbert subspaces of that given dimension are compatible. 

We then proceed to analyze the properties of these quantities, relating them to (in-)compatibility criteria. We study several examples in details, such as pairs of von Neumann measurements and mutually unbiased bases. We also provide lower bounds for these quantities using constructions inspired from the theory of error correcting codes. 

Several questions are left open. Importantly, good upper bounds on the (strong) compatibility dimensions are lacking. One would equally like to compute exactly these dimensions in very simple cases, such as the measurements in the computational basis and the one in the Fourier basis. The optimality of the algebraic techniques used in Sections \ref{sec:2-bases} (the quantity $\mathcal Z(U)$) and \ref{sec:algebraic} is also left open.  

\bigskip

\noindent\textit{Acknowledgements.} We would like to thank Andreas Bluhm, Sébastien Designolle, and Teiko Heinosaari for useful remarks and comments on a preliminary version of this paper. We would also like to thank the anonymous referee, who has helped us vastly improve the presentation of the paper. 

\bigskip

\noindent\textit{Data Availability: } data sharing not applicable --- no new data generated.
\bigskip

\printbibliography

\end{document}